\documentclass{article}

\usepackage{graphicx}
\usepackage{amsmath,amssymb}
\usepackage[hang,small,bf]{caption}
\usepackage[subrefformat=parens]{subcaption}

\newtheorem{thm}{Theorem}[section]

\newtheorem{prop}[thm]{Proposition}
\newtheorem{defn}[thm]{Definition}
\newtheorem{rem}[thm]{Remark}
\newtheorem{ass}[thm]{Assumption}

\def\l     {\left}
\def\r     {\right}
\def\<     {\langle}
\def\>     {\rangle}

\def\calB  {{\cal B}}

\def\calF  {{\cal F}}

\def\bbC   {{\mathbb C}}

\def\bbE   {{\mathbb E}}
\def\bbF   {{\mathbb F}}

\def\bbP   {{\mathbb P}}

\def\bbR   {{\mathbb R}}

\def\tP    {\bbP^\ast}

\def\tN    {\widetilde{N}}

\begin{document}
\title{On the difference between locally risk-minimizing and delta hedging strategies for exponential L\'evy models}
\author{Takuji Arai\footnote{
       Department of Economics, Keio University, 2-15-45 Mita,
       Minato-ku, Tokyo, 108-8345, Japan \ email:arai@econ.keio.ac.jp} and
       Yuto Imai\footnote{
       Department of Mathematics, Waseda University, 3-4-1 Okubo,
       Shinjyuku-ku, Tokyo, 169-8555, Japan \ email:y.imai@aoni.waseda.jp}}
\maketitle

\begin{abstract}
We discuss the difference between locally risk-minimizing and 
delta hedging strategies for exponential L\'evy models, 
where delta hedging strategies in this paper are defined under the minimal martingale measure.
We give firstly model-independent upper estimations for the difference.
In addition we show numerical examples for two typical exponential L\'evy models: Merton models and variance gamma models.
\end{abstract}

\section{Introduction}

The concept of local risk-minimization 
is widely used for contingent situations in an incomplete market framework.
Local risk-minimization is closely related to an equivalent martingale measure which is well-known as 
the minimal martingale measure (MMM).
For more details on 
local risk-minimization, see \cite{AIS} and \cite{AS}.
Delta hedging, which is also a well-known hedging method and often has been used by practitioners,
 is given by differentiating the option price
under a certain martingale measure with respect to the underlying asset price.
Due to the relationship between local risk-minimization and the MMM, we consider delta hedging under the MMM.
Its precise definition will be introduced in Section 2.

\cite{AS} showed explicit representations of local risk-minimizing (LRM) strategies for call options
by using Malliavin calculus for L\'evy processes based
on the canonical L\'evy space. 
On the other hand, Carr and Madan introduced a numerical method for valuing options based on the
fast Fourier transform (FFT) in \cite{CM}.
Carr and Madan's method was used in \cite{AIS} to compute LRM strategies of 
call options for exponential L\'evy models.
In particular, Merton models and variance gamma (VG) models were discussed
as typical examples of exponential L\'evy models.

The main motivation of this paper is to investigate whether we can use delta hedging strategies as a substitute for LRM strategies,
since we can compute delta hedging strategies much easier than LRM strategies in general.
For this purpose, we analyze the difference between the two strategies both mathematically and numerically.
First, using \cite{AIS}, we shall obtain model-independent estimations among exponential L\'evy models for the difference.
Second, in order to investigate how near the two strategies are around ``at the money'',
we provide numerical experiments for two typical exponential L\'evy models: Merton models and VG models.
Merton models are composed of a Brownian motion and compound Poisson jumps with normally distributed jump sizes.
VG models, which are exponential L\' evy processes with infinitely many jumps in any finite time interval and no Brownian component, are
the second example.

The outline of this paper is as follows:
after giving notations and preliminaries in Section 2,
we show two model-independent estimations in Section 3.
Section 4 is devoted to numerical experiments.
Conclusions are given in Section 5.
Remark that \cite{IA} treated the same problem as ours, although all results obtained in \cite{IA} are model-dependent. 
On the other hand, we obtain in this paper model-independent estimations.
In addition we shall compute numerically upper estimations of the difference between the two strategies around ``at the money.''


\section{Notations and preliminaries}

We consider a financial market composed of one risk-free asset and
one risky asset with finite maturity $T>0$.
For simplicity, we assume that market's interest rate is zero,
that is, the price of the risk-free asset is 1 at all times.
The fluctuation of the risky asset is assumed to be given by
an exponential L\'evy process $S$ on a complete probability space
$(\Omega, \mathcal{F}, \mathbb{P})$, described by
\begin{align*}
S_t:=S_0 \exp \left \{\mu t+\sigma W_t+\int_{\bbR_0}x \tN([0,t],dx)\right\}
\end{align*}
for any $t\in[0,T]$, where $S_0>0$, $\mu\in\bbR$, $\sigma>0$, 
and $\bbR_0:=\bbR\setminus\{0\}$.
Here $W$ is a one-dimensional standard Brownian motion
and $\tN$ is the compensated version of a Poisson random measure $N$.
Denoting the L\'evy measure of $N$ by $\nu$,
we have $\tN([0,t],A)=N([0,t],A)-t\nu(A)$ for any $t\in[0,T]$ and
$A\in\calB(\bbR_0)$.
Now, $(\Omega, \calF, \bbP)$ is taken as the product of a one-dimensional
Wiener space and the canonical L\'evy space for $N$.
In addition, we take $\bbF=\{\calF_t\}_{t\in[0,T]}$ as the completed canonical
filtration for $\bbP$.
For more details on the canonical L\'evy space, see \cite{S07} and \cite{AS}.
Moreover, $S$ is also a solution of the stochastic differential
equation
\[
dS_t=S_{t-}\bigg[\mu^S\,dt+\sigma \,dW_t+\int_{\bbR_0}(e^x-1)\tN(dt,dx)\bigg],
\]
where $\mu^S:=\mu+\frac{1}{2}\sigma^2 +\int_{\bbR_0}(e^x-1-x)\nu(dx)$.
Now, defining $L_t:=\log (S_t / S_{0})$ for all $t\in[0,T]$,
we have that $L$ is a L\'evy process.

Our focus is to compare LRM strategies to delta hedging strategies
for call options $(S_T-K)^+$ with strike price $K>0$.
Now, we give some preparations and assumptions.
Define the MMM $\tP$ as
an equivalent martingale measure under which any square-integrable
$\bbP$-martingale orthogonal to the martingale part of $S$.
Its density is given by
\begin{align*}
 \frac{d\tP}{d\bbP}
=\exp\big\{-\xi W_T-\frac{\xi^2}{2}T+\int_{\bbR_0}\log(1-\theta_x)N([0,T],dx)+T\int_{\bbR_0}\theta_x\nu(dx)\big\},
\end{align*}
where $\xi:=\frac{\mu^S\sigma}{\sigma^2+\int_{\bbR_0}(e^y-1)^2\nu(dy)}$
and $\theta_x:=\frac{\mu^S(e^x-1)}{\sigma^2+\int_{\bbR_0}(e^y-1)^2\nu(dy)}$
for $x\in\bbR_0$.
Remark that our discussion is strongly depending on the results in \cite{AIS}.
Thus, we need the assumptions imposed in \cite{AIS} as follows:
\begin{ass}\label{ass-1}
\begin{enumerate}
\item
      $\int_{\bbR_0}(|x|\vee x^2)\nu(dx)<\infty$\textnormal{,} and
      $\int_{\bbR_0}(e^x-1)^n\nu(dx)<\infty$ for $n=2,4$. \vspace{2pt}
\item $0\geq\mu^S>-\sigma^2-\int_{\bbR_0}(e^x-1)^2\nu(dx)$.
\end{enumerate}
\end{ass}
The first condition ensures
that (i) $\mu^S$, $\xi$, and $\theta_x$
are well defined, (ii) $L$ is square integrable, and (iii)
$\int_{\bbR_0}(e^x-1)^n\nu(dx) < \infty$ for $n=1,3$.
The second condition guarantees that $\theta_x<1$ for any $x\in\bbR_0$.
Now we consider
\begin{align}
\bbE_{\tP}[{\bf 1}_{\{S_T>K\}}S_T\mid\calF_{t-}]\,, \label{I1}
\end{align}
and
\begin{align}
\int_{\bbR_0}\bbE_{\tP}[(S_Te^x-K)^+ -(S_T-K)^+\mid\calF_{t-}] (e^x-1)\nu(dx)\, . \label{I2} 
\end{align}
Noting that (\ref{I1}) and (\ref{I2}) are functions of $S_{t-}$ and $K$, we denote them by $I_{1}(S_{t-}, K)$ and $I_{2}(S_{t-}, K)$, respectively.
LRM strategies are given as a predictable process $LRM(S_{t-}, K )$,
which represents the number of units of the risky asset the investor holds
at time $t$.
Here
its explicit representation for call options $(S_T-K)^+$ is given as follows:
\begin{prop}[Proposition 4.6 of \cite{AS}]
For any $K>0$ and $t\in[0,T]$, 
\begin{align}
\label{eq-prop-AS}
LRM(S_{t-}, K )=\frac{\sigma^2 I_1(S_{t-},K) + I_2(S_{t-},K)}{S_{t-}\big(\sigma^2
      + C_2 \big)}.
\end{align}
where $C_2 := \int_{\bbR_0}(e^x-1)^2\nu(dx)$.
 
\end{prop}

In addition, we introduce integral representations given in \cite{AIS} for $I_1(S_{t-},K)$ and $I_2(S_{t-},K)$ in order to see that Carr and Madan's method is available.
The characteristic function of $L_{T-t}$ under $\tP$ is denoted by 
$\phi_{T-t}(z):=\bbE_{\tP}[e^{izL_{T-t}}]$ for $z\in\bbC$.
We induce an integral representation for
$I_1(S_{t-},K)$ with $\phi_{T-t}$ firstly as follows:
\begin{align*}
I_1(S_{t-},K)&=\bbE_{\tP}[{\bf 1}_{\{S_T>K\}}\cdot S_T\mid\calF_{t-}] \nonumber \\
&= \frac{1}{\pi}\int_0^\infty\frac{K^{-\alpha + 1 -iv}}{\alpha-1+iv}
   \phi_{T-t}(v-i\alpha)S_{t-}^{\alpha+iv}dv \nonumber \\
&= \frac{e^k}{\pi}\int_0^\infty e^{-i(v-i\alpha)k}\psi_1(v-i\alpha)dv\, ,
\end{align*}
where $k:=\log K$ and
$\psi_1(z):=\frac{\phi_{T-t}(z)S^{iz}_{t-}}{-1 + iz}$ and $\alpha \in (1, 2]$.
Note that the right-hand side is independent of the choice of $\alpha$.
We turn next to $I_2(S_{t-},K)$.
Denoting $\psi_2(z):=\frac{\phi_{T-t}(z)S_{t-}^{iz}}{(-1 + iz)iz}$, we have
\begin{align*}
I_2(S_{t-},K)
&= \int_{\bbR_0}\bbE_{\tP}[(S_Te^x-K)^+-(S_T-K)^+\mid\calF_{t-}] (e^x-1)\nu(dx) \nonumber\\
&= \frac{1}{\pi}\int_0^\infty K^{-\alpha + 1 - iv}
   \int_{\bbR_0}(e^{(\alpha + iv) x}-1)(e^x-1)\nu(dx)\psi_2(v -i \alpha)dv\,.
\end{align*} 
Regarding $LRM(S_{t-}, K)$, $I_1(S_{t-},K)$, and $I_2(S_{t-},K)$ as functions of $S_{t-}$ and $K$,
we have $I_j(S_{t-},K)/S_{t-}=I_j(1,K/S_{t-})$ for $j=1,2$, and
\[
LRM(S_{t-},K)
=\frac{\sigma^2I_1(1,K/S_{t-})+I_2(1,K/S_{t-})}{
  \sigma^2+ C_2}
\]
from (\ref{eq-prop-AS}).
As a result, $LRM(S_{t-}, K)$ is given as a function of $\chi_{t-} := K/S_{t-}$,
where $\chi_{t-}$ is called \textit{moneyness}.
Thus we denote $LRM(S_{t-}, K)$ by $LRM(\chi_{t-})$.
 
 Next, we define delta hedging strategies.
\begin{defn}
For any $K > 0$ and $s > 0$, 
a delta hedging strategy $\Delta(S_{t-},K)$ under $\tP$ for a call option with strike price $K$ is defined as
\begin{align*}
\Delta(S_{t-},K) := \frac{\partial \bbE_{\tP}[(S_T -K)^+ \mid S_{t-} = s]}{\partial s} \, .
\end{align*}
\end{defn}
Remark that the above definition of delta hedging strategies
coincides with that of usual delta hedging strategies in the Black--Scholes model.
The next theorem follows from a direct calculation.
\begin{thm} 
We have 
\begin{align*}
\Delta(S_{t-},K) &= \frac{I_1(S_{t-},K)}{S_{t-}} \,\, .
\end{align*} 
\end{thm}
Note that $\Delta(S_{t-},K)$ is given as a function of $\chi_{t-}$ also.
Thus we denote $\Delta(S_{t-},K)$ by $\Delta(\chi_{t-})$.

\begin{rem}
\cite{DGK} studied similar problems to this paper.
They compared some hedging errors among variance-optimal hedge, Black-Scholes hedge, and delta hedge.
\end{rem}

\section{Main results}
We give two estimations of the difference $|LRM(\chi_{t-}) - \Delta(\chi_{t-})|$ as main results of this paper. 
Remark that the estimations given in this section are independent of any exponential L\'evy models.
Throughout this section we fix $t \in [0, T]$ arbitrary. 
We denote $\chi :=\chi_{t-}$ for short, and regard LRM and $\Delta$ as functions of $\chi \in \mathbb{R}^{+}$.
Let $p^\ast$ be the distribution of $L_{T-t}$ under $\tP$, that is, $p^\ast(A) := \tP(L_{T-t} \in A)$ for any $A \in \calB(\bbR_0)$.

First we give an estimation of $|LRM(\chi) - \Delta(\chi)|$, which is useful when $\chi > 0$ is small.
\begin{thm} \label{thm3}
For any $\chi \in \mathbb{R}^+$, we have the following inequality estimation:
\begin{align}
|LRM(\chi) - \Delta(\chi)| \leq \frac{\chi C^-_2}{\sigma^2 + C_2} + \frac{\chi p^\ast((-\infty, \log \chi])}{\sigma^2 + C_2} (C^+_2 - C^-_2) \, , \label{est}
\end{align}
where
\begin{align*}
C_2^+ := \int^{\infty}_{0} (e^x - 1)^2 \nu(dx) \, , \quad \mbox{and} \quad
C_2^- := \int^{0}_{-\infty} (e^x - 1)^2 \nu(dx) \, . 
\end{align*}
Hence we have $|LRM(\chi) - \Delta(\chi)| \leq \mathcal{O}(\chi)$ as $\chi \to 0$.
\end{thm}

\begin{proof}
We denote $I_{1}(1, \chi)$ and $I_{2}(1, \chi)$ by $I_{1}$ and $I_{2}$ for short.
First of all, we decompose $I_2$ into $I_{2} = J_{1} + J_{2}+ J_{3} + J_{4}$.
Here
\begin{align*}
J_{i} := \int_{D_i} \{(e^{y+x} - \chi)^{+} - (e^{y} - \chi)^{+} \} (e^{x} - 1) p^\ast(dy)\nu(dx),
 \quad i = 1, \cdots, 4 \, ,
\end{align*}
where 
\begin{align*}
D_1 := \{ (x,y) | x + y \geq \log \chi , y \geq \log \chi  \} \, , \quad
D_2 := \{ (x,y) | x + y \geq \log \chi , y < \log \chi  \} \, , \\
D_3 := \{ (x,y) | x + y < \log \chi , y \geq \log \chi  \} \, , \quad
D_4 := \{ (x,y) | x + y < \log \chi , y < \log \chi  \} \, . \\
\end{align*}
Thus we have
\begin{align}
|LRM(\chi) - \Delta(\chi)|
&= \l| I_1 - \frac{\sigma^2 I_1 + I_2}{\sigma^2 + C_2} \r| \nonumber \\
&= \frac{1}{\sigma^2 + C_2} \l| C_2 I_1 -J_1 -J_2 - J_3 - J_4 \r| \, . \label{i}
\end{align}
Noting that $J_{4} = 0$ and 
\begin{align*}
C_2 I_1
&= \int_{D_1 \cup D_3} e^y (e^x - 1)^2 p^\ast(dy)\nu(dx) \, ,
\end{align*}
we obtain
\begin{align}
C_2 I_1 - J_1 - J_3 - J_{4}
&= \int_{D_3} \l\{ e^y(e^x -1)^2 + (e^y - \chi)(e^x -1) \r\} p^\ast(dy)\nu(dx) \nonumber \\
&= \int_{D_3} ( e^{y+x} - \chi )(e^x -1) p^\ast(dy)\nu(dx) \nonumber \\
&= \int^0_{-\infty} \int^{\log \chi - x}_{\log \chi} (e^{y+x} - \chi)(e^x - 1) p^\ast(dy)\nu(dx) \nonumber \\
&\leq \int^0_{-\infty} \int^{\log \chi - x}_{\log \chi} \chi (e^x - 1)^2 p^\ast(dy)\nu(dx) \nonumber \\
&\leq \chi p^\ast([\log \chi, \infty))C^-_{2} \, . \label{ro}
\end{align}
In the same manner, we have
\begin{align}
J_2
&= \int^{\infty}_0 \int^{\log \chi }_{\log \chi - x} (e^{y+x} - \chi)(e^x - 1) p^\ast(dy)\nu(dx) \nonumber \\
&\leq \int^{\infty}_0 \int^{\log \chi }_{\log \chi -x} \chi (e^x - 1)^2 p^\ast(dy)\nu(dx) \label{A} \\
&\leq \chi p^\ast(( -\infty, \log \chi])C^+_{2} \, . \label{ha}
\end{align}
From (\ref{i}), (\ref{ro}), and (\ref{ha}), we can conclude
\begin{align}
|LRM(\chi) - \Delta(\chi)| 
&\leq \frac{1}{\sigma^2 + C_2} \{ |\chi p^\ast ([\log \chi, \infty))C^-_2| +|J_2| \} \label{I}\\
&\leq \frac{\chi}{\sigma^2 + C_2} \l\{ p^\ast([\log \chi, \infty))C^-_{2}  +  p^\ast(( -\infty, \log \chi])C^+_{2}  \r\} \nonumber \\
&= \frac{\chi C^-_2}{\sigma^2 + C_2} + \frac{\chi p^\ast((-\infty, \log \chi])}{\sigma^2 + C_2} (C^+_2 - C^-_2)\, . \nonumber 
\end{align}
This completes the proof of Theorem \ref{thm3}.
\end{proof}

Next we give the second estimation of $|LRM(\chi) - \Delta(\chi)|$ for large $\chi$.
\begin{thm}
Suppose
\begin{align}
\int^\infty_0 \frac{|\phi_{T-t}(v - 2i)|}{1 + v} dv  <\infty \, . \label{ass1}
\end{align}
Then there exists a constant $\mathcal{C} > 0$ such that 
\begin{align}
|LRM(\chi) - \Delta(\chi)| \leq \frac{\mathcal{C}}{\chi} \label{est2}
\end{align}
for any $\chi \in \mathbb{R}^{+}$. 
So that, $|LRM(\chi) - \Delta(\chi)| \leq \mathcal{O} (\frac{1}{\chi})$ as $\chi \to \infty$.
\end{thm}
\begin{proof}
We show (\ref{est2}) by using (\ref{I}).
To this end we estimate $p^\ast ([\log \chi, \infty))$ and $J_{2}$ separately.

In order to estimate $p^\ast ([\log \chi, \infty))$,
we define a function $\widehat{g}_1$ as
\begin{align*}
\widehat{g}_1(z)
:= \int_{\bbR} e^{izx} \mathbf{1}_{[\log \chi, \infty)} (x) dx = -\frac{e^{iz \log \chi}}{iz} \, 
\end{align*}
for $z \in \mathbb{C}$, which implies that
\begin{align}
p^\ast ([\log \chi, \infty)) 
&= \int^\infty_{-\infty} \mathbf{1}_{[\log \chi, \infty)} (x) p^\ast (dx) \, \nonumber \\ 
&= \frac{1}{\pi} \int^\infty_{0} \widehat{g}_1(-v + i \alpha) \phi_{T-t}(v - i\alpha) dv \nonumber \\
&= \frac{1}{\pi} \int^\infty_{0} \frac{\chi^{-\alpha -iv}}{\alpha + iv} \phi_{T-t}(v - i\alpha) dv \, , \label{star1}
\end{align}
where $\alpha \in (1,2]$ and the value of (\ref{star1}) is independent of the choice of $\alpha$.
Remark that the second equation of (\ref{star1}) is from (2.17) in \cite{AIS}.
We may choose $\alpha = 2$ without loss of generality.
Hence we have
\begin{align*}
p^\ast ([\log \chi, \infty)) 
&= \frac{1}{\pi} \int^\infty_{0} \frac{\chi^{-2 -iv }}{2 +iv} \phi_{T-t}(v - 2i) dv\\
&\leq \frac{1}{\pi} \frac{1}{\chi^2} \int^\infty_{0} |\chi^{-iv}| \frac{|\phi_{T-t}(v - 2i)|}{|2 + iv|} dv \\
&= \frac{1}{\pi} \frac{1}{\chi^2} \int^\infty_0 \frac{|\phi_{T-t}(v - 2i)|}{1 + v} \l| \frac{1 + v}{2 + iv} \r| dv \, .
\end{align*}
Denoting $f(v) := \l| \frac{1 + v}{2 + iv} \r| = \frac{1 + v}{\sqrt{4 + v^{2}}}$ for $v \geq 0$,  we can see that $\frac{1 + v}{\sqrt{4 + v^{2}}} \leq \frac{\sqrt{5}}{2}$ for any $v \geq 0$.
From (\ref{ass1}),  we have
\begin{align}
p^\ast ([\log \chi, \infty)) 
&\leq \frac{\sqrt{5}}{2 \pi \chi^2} \int^\infty_0 \frac{|\phi_{T-t}(v - 2i)|}{1 + v} dv < \infty \, . \label{ine1}
\end{align}

Next we check the $J_2$ part. 
(\ref{A}) implies
\begin{align*}
J_2
&\leq \int^{\infty}_0 \int^{\log \chi}_{\log \chi - x} \chi (e^x -1)^2 p^\ast(dy) \nu(dx) \\
&\leq \chi \int^{\infty}_0 p^\ast([\log \chi - x, \infty))   (e^x -1)^2  \nu(dx) \, .
%
\end{align*}
In the same manner as the above estimation for $p^\ast ([\log \chi, \infty))$, we estimate $p^\ast([\log \chi - x, \infty))$ by using (\ref{star1}).
Replacing 
$\log \chi$ with $\log \chi - x$ and substituting 2 for $\alpha$, 
we have
\begin{align}
p^\ast ([\log \chi - x, \infty)) 
&= \frac{1}{\pi} \int^\infty_{0} \frac{\chi^{-2 -iv} e^{(2 + iv)x}}{2 + iv} \phi_{T-t}(v - 2 i) dv \, . \nonumber 
\end{align}
Hence we have
\begin{align}
J_2
&\leq \chi \int^{\infty}_0 p^\ast([\log \chi - x, \infty)) (e^x -1)^2  \nu(dx) \nonumber \\
&= \frac{\chi}{\pi} \int^{\infty}_0 \int^\infty_{0} \frac{\chi^{-2 -iv} e^{(2 + iv)x}}{2 + iv} \phi_{T-t}(v - 2i)  (e^x -1)^2 dv  \nu(dx) \nonumber \\
&\leq \frac{1}{\pi} \frac{1}{\chi} \int^{\infty}_0 \int^\infty_{0} \l| \frac{\chi^{-iv} e^{ivx}}{2 + iv } \phi_{T-t}(v - 2i) \r| e^{2x} (e^x -1)^2 dv  \nu(dx) \nonumber \\
&\leq \frac{\sqrt{5}}{2 \pi \chi} \int^\infty_0 \frac{|\phi_{T-t}(v - 2i)|}{1 + v} dv \int^{\infty}_0 e^{2x} (e^x -1)^2 \nu(dx) \, . \label{B}
\end{align}
Noting that (\ref{ass1}), and $ \int^{\infty}_0 e^{2x} (e^x -1)^2 \nu(dx) < \infty$ from Assumption \ref{ass-1},
we obtain $|J_2| < \infty.$

From (\ref{I}), (\ref{ine1}), and (\ref{B}) we obtain 
\begin{align*}
|LRM(\chi) - \Delta(\chi)| \leq \frac{\sqrt{5}}{2 \pi (\sigma^{2} + C_{2}) \chi}
&  \int^{\infty}_{0} \frac{|\phi_{T-t}(v - 2i)|}{1 + v} dv  \nonumber \\
& \times \l\{C_{2}^{-} +   \int^{\infty}_{0} e^{2x}(e^{x} - 1)^{2} \nu(dx)  \r\} . 
\end{align*}
\end{proof}

\begin{rem}
The condition (\ref{ass1}) is not necessarily satisfied for the case of $\sigma =0$, although it holds whenever $\sigma>0$.
Thus, the proof of Proposition 2.1 in \cite{AIS} includes an error.
On the other hand, \cite{AIS} treated only Merton and VG models, and we can see (\ref{ass1}) for both models,
because $\sigma > 0$ for Merton models, and Proposition 4.7 in \cite{AIS} for VG models.
\end{rem}

\section{Numerical results}


In this section, we implement numerical experiments for two typical exponential L\'evy models, known as Merton models and VG models.
We obtain in Section 3 two estimations for the difference $|LRM(\chi)-\Delta(\chi)|$, and see that the difference converges to $0$ as $\chi$ tends to $0$ or $\infty$.
On the other hand, we are interested in the behaviour of the difference around ``at the money" from the practical point of view.
To investigate it, we compute the values of the right-hand side of (\ref{est}) in Theorem 3 as an upper estimation, and compare them with the values of the difference.
Remark that the numerical scheme developed in this section is based on the results of \cite{AIS}.


\subsection{The Merton jump-diffusion models}

We consider the case where $L = \log(S / S_{0})$ is given
as a Merton jump-diffusion process,
which consists of a diffusion component with volatility $\sigma>0$ and
compound Poisson jumps with three parameters, $m\in\bbR$, $\delta>0$, and
$\gamma>0$.
Note that $\gamma$ represents the jump intensity, and
the sizes of the jumps are distributed normally with mean $m$ and
variance $\delta^2$.
Thus, its L\'evy measure $\nu$ is given by
\begin{align*}
\nu(dx)
=\frac{\gamma}{\sqrt{2\pi}\delta}\exp\left\{-\frac{(x-m)^2}{2\delta^2}\right\}dx \, . \label{MertonNu}
\end{align*}
Note that the first condition of Assumption~\ref{ass-1} is satisfied
for any $m\in\bbR$, $\delta>0$, and $\gamma>0$.
Thus we consider only parameter sets satisfying the second condition of
 Assumption~\ref{ass-1}.
An analytic form of $\phi_{T-t}$ was given in Proposition 3.1 of \cite{AIS}.
%
We compute the right-hand side of (\ref{est}) in Theorem 3 with FFT.
Note that the constant $C_{2}^{-}$ is given as follows:
\begin{align*}
C_2^- = 
\gamma \l[ e^{2(\delta^{2} + m)} \Phi \l( -\frac{2 \delta^{2} + m}{\delta} \r)
	-2 e^{\frac{\delta^{2} + 2m}{2}} \Phi \l( -\frac{\delta^{2} + m}{\delta} \r)
	+\Phi \l( -\frac{m}{\delta} \r)
 \r] \, ,
\end{align*}
where $\Phi$ is the standartd normal cumulative distribution function.
Moreover, $p^*((-\infty, \log\chi])$ is calculated with FFT as follows:
\begin{align*}
p^*((-\infty, \log\chi]) = 
1 - \frac{1}{\pi}\int^\infty_0 \chi^{-\alpha - i v} \frac{\phi_{T-t}(v - i \alpha)}{\alpha + i v} dv \, .
\end{align*}

\subsection{The variance gamma models}
Next we consider the case where
$L$ is given as a variance gamma process with three parameters $\kappa>0$, $m\in\bbR$, and $\delta>0$, which is defined as a time-changed Brownian motion with volatility $\delta$,
drift $m$, and subordinator $G_t$, where $G_t$ is a gamma process
with parameters $(1/\kappa, 1/\kappa)$.
In summary, $L$ is represented as
\begin{align*}
L_t=mG_t+\delta B_{G_t}\ \ \mbox{ for }t\in[0,T]\,,
\end{align*}
where $B$ is a one-dimensional standard Brownian motion.
Moreover, the L\'evy measure of $L$ is given by
\begin{align*}
\nu(dx)
= C(\mathbf{1}_{\{x<0\}}e^{-G|x|}+\mathbf{1}_{\{x>0\}}e^{-M|x|})\frac{dx}{|x|} \, ,
\end{align*}
where
\begin{align*}
C :=\frac{1}{\kappa} >0, \quad
G :=\frac{1}{\delta^2}\sqrt{m^2+\frac{2\delta^2}{\kappa}}+\frac{m}{\delta^2} >0, 
\quad
M :=\frac{1}{\delta^2}\sqrt{m^2+\frac{2\delta^2}{\kappa}}-\frac{m}{\delta^2} > 0. 
\end{align*}
In addition, we assume $M>4$, which ensures
the first condition of Assumption~\ref{ass-1}.
An analytic form of $\phi_{T-t}$ was given in Proposition 4.5 of \cite{AIS}.
In order to compute the right-hand side of (\ref{est}),
we calculate the constants $C_{2}^{+}$ and $C_{2}^{+}$ explicitly as follows:
\begin{align*}
C_2^+ = C \log\l(\frac{(M - 1)^2}{M (M - 2)} \r) \quad \mbox{and} \quad
C_2^- = C \log\l(\frac{(G + 1)^2}{G (G + 2)} \r).
\end{align*}
In the same manner as Merton models, $p^*((-\infty, \log\chi])$ is calculated with FFT as follows:
\begin{align*}
p^*((-\infty, \log\chi]) = 
1 - \frac{1}{\pi}\int^\infty_0 \chi^{-\alpha - i v} \frac{\phi_{T-t}(v - i \alpha)}{\alpha + i v} dv \, .
\end{align*}

\subsection{Numerical methods, data, and results}

We calibrate models' parameter sets against a set of European call options on the S\&P 500 Index.
Note that models' prices are defined as the expected value under the MMM $\tP$.
The data set consists of 81 mid-prices at the close of the market on 20 April 2016.
On the day, the S\&P 500 Index closed at 2102.4.
The markets' prices consist of seven expirations, which are in 20 May 2016, 17 June 2016, 15 July 2016, 16 September 2016, 16 December 2016,
20 January 2017, and finally 17 March 2017.  
We calibrate Merton and VG models to this data set.
We compute the root-mean-squared error (RMSE) between the market's and model's prices.
This statistic is an measure of the quality of fit.
We estimate model parameter sets by minimizing RMSE via SQP method.

We provide calibration results for Merton and a VG models.
The estimated parameter set for the Merton case is 
$\mu=4.0073$, $\sigma=0.0435$, $\gamma=0.0054$, $m=-0.0697$, and $\delta=0.0889$.
Under this parameter set, RMSE is $3.7809$.
The above estimated parameter set satisfies the second condition of
Assumption~\ref{ass-1}.
We set $T =1$, $t = 0.95$, and $S_{t} = 2102.4$, respectively.
In Figure \ref{fig-Merton} (a), the values of $LRM(\chi)$ and $\Delta(\chi)$ are plotted separately for $K=1900, 1950, \dots, 2500$ $(\chi = 0.9037, 0.9275, \cdots, 1.1891)$.
The values of $|LRM(\chi) - \Delta(\chi)|$ and their upper estimations using (\ref{est}) are seen as Figure \ref{fig-Merton} (b).
FFT parameters are chosen as $N = 2^{14}$, $\eta = 0.025$ and $\alpha = 1.75$.
%
As for the VG case, the estimated parameter set is given as follows:
$C = 6.7910$, 
$G = 30.1807$, and 
$M = 33.1507$.
Under this parameter set, RMSE is $6.429$.
This parameter set also satisfies the second condition of Assumption \ref{ass-1}.
We implement the same numerical experiments as the Merton case.
%
Figure \ref{fig-VG} shows their results.

We deduce three points from our numerical experiments:
(i) The upper estimate from (\ref{est}) fits very well into the real values of $|LRM(\chi) - \Delta(\chi)|$ for the Merton case.
(ii) The values of $|LRM(\chi) - \Delta(\chi)|$ 
for the VG case are larger than ones for the Merton case.
(iii) For the VG case, the behaviour of $|LRM(\chi) - \Delta(\chi)|$ 
is unstable around ``at the money''.

\section{Conclusions}
We derive inequality estimations for $|LRM(\chi) - \Delta(\chi)|$ in Theorems 3 and 4.
In particular, we show the difference converges to $0$ with order less than $\mathcal{O}(\chi)$ and $\mathcal{O}(1/\chi)$ when $\chi$ tends to $0$ and $\infty$, respectively.
We compute the behaviour of $|LRM(\chi) - \Delta(\chi)|$ for two models: Merton and VG models.
For any exponential L\'evy model, computing the right-hand side of (\ref{est}) in Theorem 3 is a simple way to evaluate roughly the distance between the two strategies.
In particular, Figure 1 shows that it is appropriate to use delta hedging strategies as a substitute for LRM strategies.
On the other hand, Theorem 4 gives an estimation for large $\chi$, and seems not to be useful to evaluate around ``at the money".


\begin{figure}[htb]
\vspace{-3.5cm}
  \begin{minipage}[t]{1\textwidth}
    \centering
\hspace{-1cm}
    \includegraphics[keepaspectratio, scale=0.5]{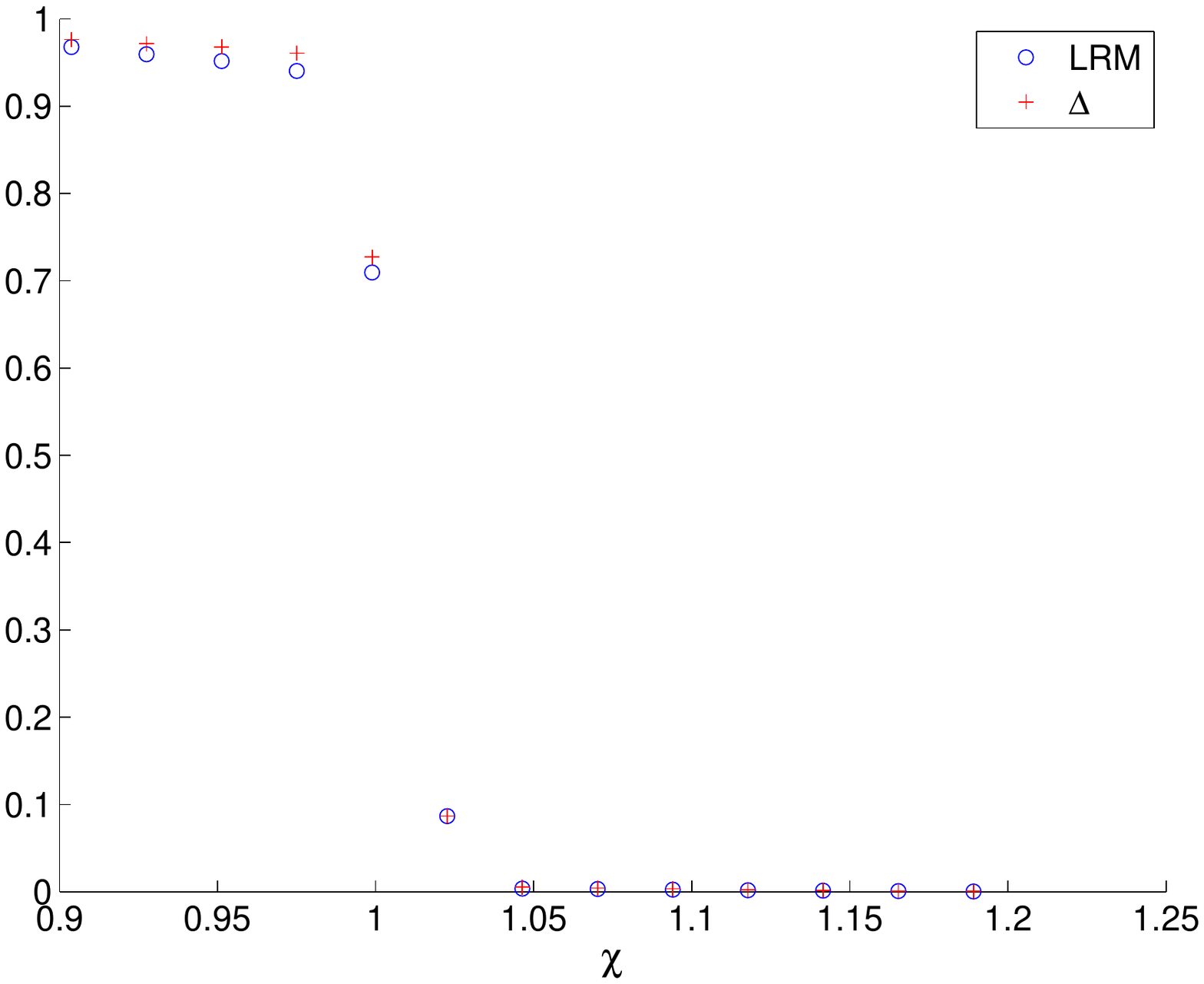} \vspace{-3cm}
    \subcaption{Values of $LRM(\chi)$ and $\Delta(\chi)$ when $t$ is fixed to 0.95 vs. moneyness $\chi$ from 0.9037 to 1.1891 (K from 1900 to 2500). Blue circles and red crosses  represent the values of $LRM(\chi)$ and $\Delta(\chi)$, respectively. }
  \end{minipage}
\\
    \begin{minipage}[t]{1\textwidth}
   \centering
 \vspace{-5cm}
    \includegraphics[keepaspectratio, scale=0.5]{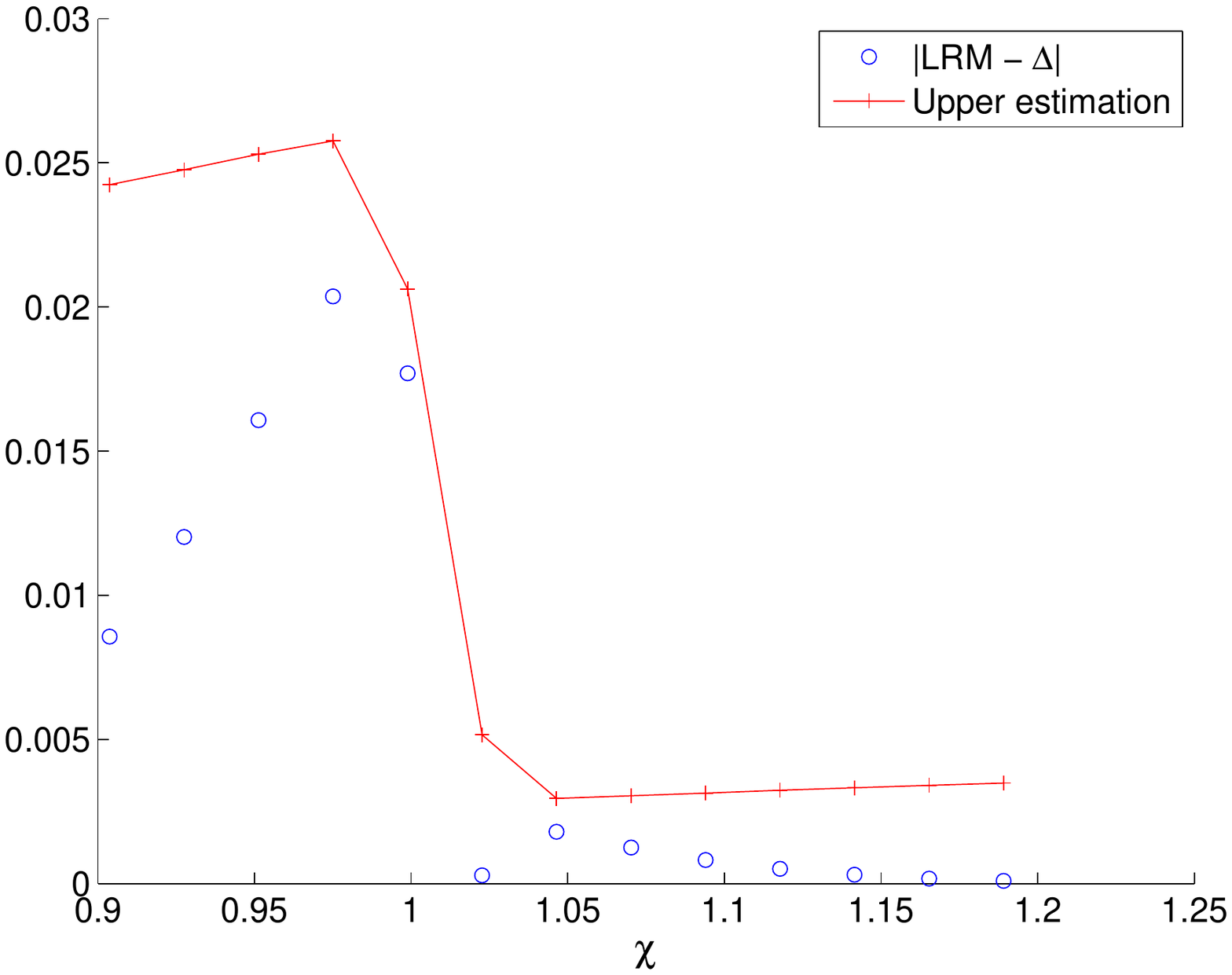} \vspace{-3cm}
    \subcaption{ Values of $|LRM(\chi) -\Delta(\chi)|$ and upper estimations when $t$ is fixed to 0.95 vs. moneyness $\chi$ from 0.9037 to 1.1891 (K from 1900 to 2500). Blue circles and red crosses  represent the values of $|LRM(\chi) -\Delta(\chi)|$ and upper estimations, respectively. }
  \end{minipage} 
  \caption{Merton model with the estimated parameters, $\mu=4.0073$, $\sigma=0.0435$, $\gamma=0.0054$, $m=-0.0697$, and $\delta=0.0889$
  } \label{fig-Merton}
\end{figure}


\begin{figure}[htb]
\vspace{-3.5cm}
  \begin{minipage}[t]{1\textwidth}
    \centering
    \includegraphics[keepaspectratio, scale=0.5]{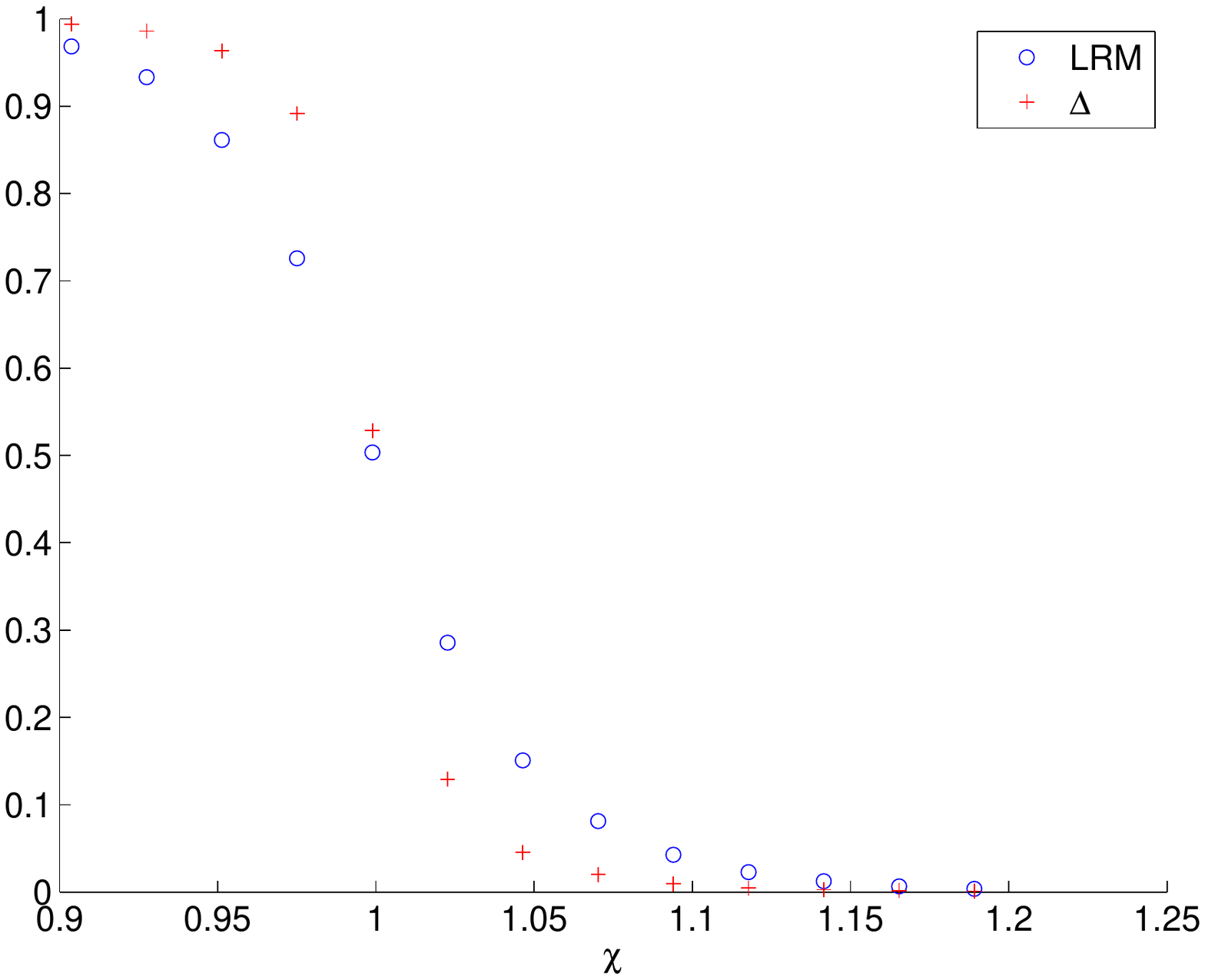} \vspace{-3cm}
    \subcaption{Values of $LRM(\chi)$ and $\Delta(\chi)$ when $t$ is fixed to 0.95 vs. moneyness $\chi$ from 0.9037 to 1.1891 (K from 1900 to 2500).}
  \end{minipage}
\\
  \begin{minipage}[t]{1\textwidth}
    \vspace{-5cm}
    \centering
    \includegraphics[keepaspectratio, scale=0.5]{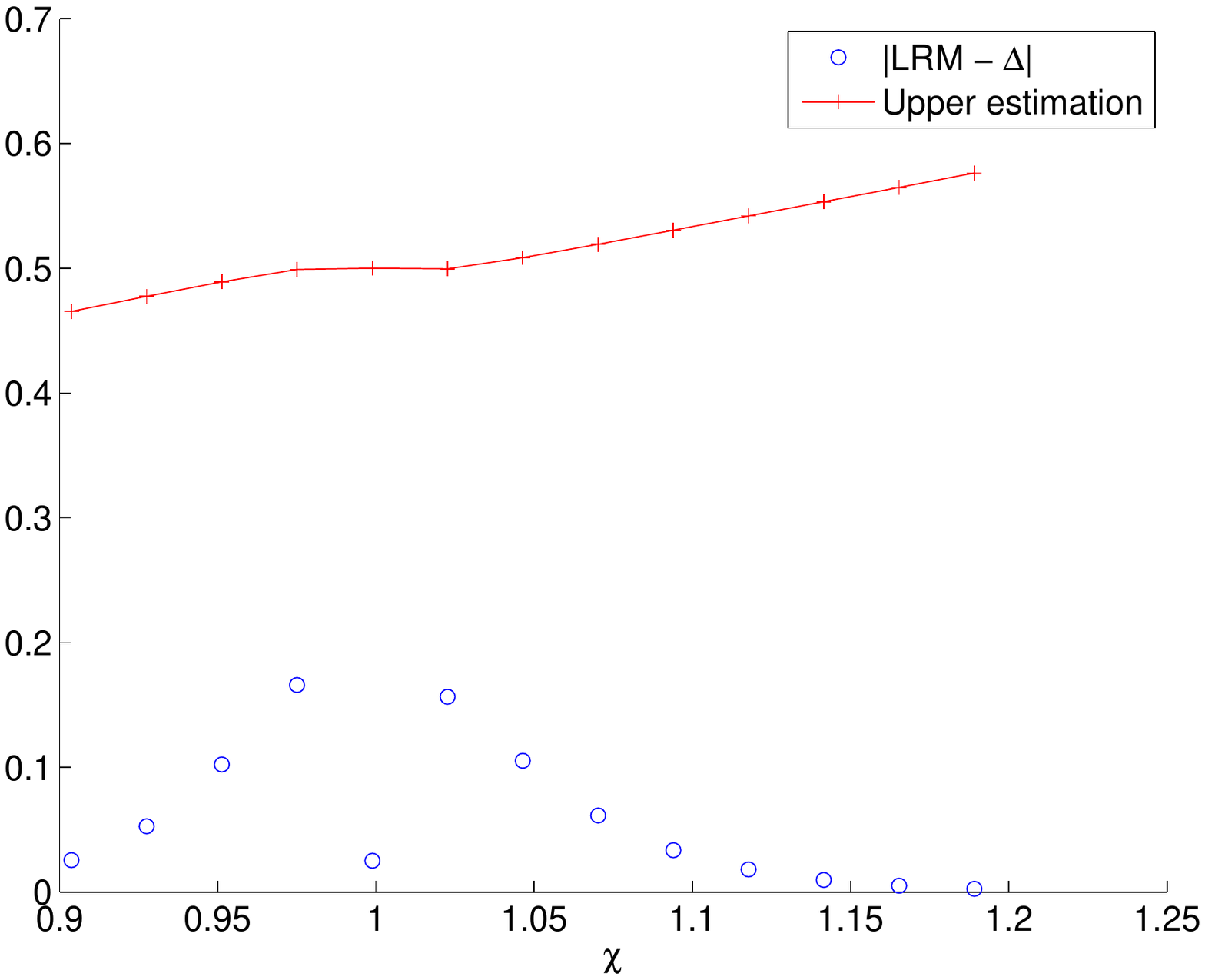} \vspace{-3cm}
    \subcaption{Values of $|LRM(\chi) -\Delta(\chi)|$ and upper estimations when $t$ is fixed to 0.95 vs. moneyness $\chi$ from 0.9037 to 1.1891 (K from 1900 to 2500).}
  \end{minipage} 
  \caption{VG model with estimated parameters, $C = 6.7910$, 
$G = 30.1807$, and 
$M = 33.1507$
  } \label{fig-VG}
\end{figure}


\begin{center}
{\bf Acknowledgements}
\end{center}
Takuji Arai was supported by JSPS Grant-in-Aid for Scientific Research (C) No.15K04936.
Yuto Imai was supported by Waseda University Grants for Special Research Projects (Project number: 2016K-174 and 2016B-123).


\begin{thebibliography}{99}
\bibitem[1]{AIS}
   T. Arai, Y. Imai and R. Suzuki,
   Numerical local risk minimization for exponential L\'evy models:
   {\it Int. J. Theor. Appl. Finan.}, {\bf 19}, 1650008 (2016)
\bibitem[2]{AS}
   T. Arai and R. Suzuki,
   Local risk-minimization for L\'evy markets.
   {\it Int. J. Finan. Eng.}, {\bf 02}, 1550015 (2015).
\bibitem[3]{CM}
   P. Carr and D. Madan, 
   Option valuation using the fast Fourier transform,
    {\it J. Comp. Finan.}, {\bf 2} (1999), 61--73.
\bibitem[4]{DGK}
S. Denkl, M. Goy, J. Kallsen, J. Muhle-Karbe and A. Pauwels, 
 On the performance of delta hedging strategies in exponential L\'evy models, 
 {\it Quant. Finan.}, {\bf 13} (2013), 1173--1184.
\bibitem[5]{IA} Y. Imai and T. Arai,
Comparison of Local Risk Minimization and Delta Hedging for Exponential L\'evy Models,
{\it JSIAM Letters}, {\bf 7} (2015), 77 -- 80. 
\bibitem[6]{S07}
   J. L. Sol\'e, F. Utzet, J. Vives 
   Canonical L\'evy process and Malliavin calculus,
   {\it Stochastic Processes and their Applications} {\bf 117} (2007), 165--187.
%
%
\end{thebibliography}
\end{document}